\documentclass [11pt] {article}

 \usepackage{fullpage}

\usepackage{graphics}
\usepackage[dvips]{epsfig}

\usepackage{amsmath}
\usepackage{amssymb}
\usepackage{amsfonts}
\usepackage{graphicx}

\usepackage{cite}

\usepackage{algorithmic}
\usepackage[linesnumbered,ruled]{algorithm2e}

\usepackage{array}
\usepackage{color}

\usepackage{multirow}
\usepackage{rotating}

\begin{document}

\newtheorem{theorem}{Theorem}[section]
\newtheorem{lemma}{Lemma}[section]
\newtheorem{corollary}{Corollary}[section]
\newtheorem{claim}{Claim}[section]
\newtheorem{proposition}{Proposition}[section]
\newtheorem{definition}{Definition}[section]
\newtheorem{fact}{Fact}[section]
\newtheorem{example}{Example}[section]

\newcommand{\cA}{{\cal A}}
\newcommand{\cB}{{\cal B}}
\newcommand{\cC}{{\cal C}}
\newcommand{\cG}{{\cal G}}
\newcommand{\cN}{{\cal N}}
\newcommand{\cU}{{\cal U}}
\newcommand{\cT}{{\cal T}}
\newcommand{\cS}{{\cal S}}
\newcommand{\cL}{{\cal L}}
\newcommand{\cV}{{\cal V}}
\newcommand{\loc}{{\cal LOCAL}}

\newcommand{\ai}{\alpha_i}
\newcommand{\bi}{\beta_i}
\newcommand{\gi}{\gamma_i}
\newcommand{\di}{\delta_i}

\newcommand{\oai}{\overline{\alpha}_i}
\newcommand{\obi}{\overline{\beta}_i}
\newcommand{\ogi}{\overline{\gamma}_i}
\newcommand{\odi}{\overline{\delta}_i}

\newcommand{\qed}{\hfill $\square$ \smallbreak}
\newenvironment{proof}{\noindent{\bf Proof:}}{\qed}

\newcommand{\procend}{\hfill $\diamond$\medskip}

\newcommand{\oddRepair}{{\tt Odd\-Repair}}
\newcommand{\Deactivate}{{\tt Deactivate}}
\newcommand{\evenARepair}{{\tt Even\-Al\-most\-Re\-pair}}
\newcommand{\ringThree}{{\tt Ring\-Three\-Co\-lo\-ring}}
\newcommand{\ringLearning}{{\tt Ring\-Lear\-ning}}
\newcommand{\Elect}{{\tt Elect}}



\def\thefootnote{\fnsymbol{footnote}}

\title{{\bf Latecomers Help to Meet:\\
 Deterministic Anonymous Gathering in the Plane}}

\author{Andrzej Pelc\footnotemark[1]
\and
Ram Narayan Yadav\footnotemark[2]
}

\footnotetext[1]{
 D\'epartement d'informatique, Universit\'e du Qu\'ebec en Outaouais, Gatineau,
Qu\'ebec J8X 3X7, Canada. {\tt pelc@uqo.ca}. Partially supported by NSERC discovery grant 2018-03899
and by the Research Chair in Distributed Computing at the
Universit\'e du Qu\'ebec en Outaouais.}

\footnotetext[2]{D\'epartement d'informatique, Universit\'e du Qu\'ebec en Outaouais, Gatineau,
Qu\'ebec J8X 3X7, Canada. {\tt narayanram.1988@gmail.com}}

\maketitle

\thispagestyle{empty}

\begin{abstract}
A team of anonymous mobile agents represented by points freely moving in the plane have to gather at a single point and stop. Agents start at
different points of the plane and at possibly different times chosen by the adversary. They are equipped with compasses, a common unit of distance and clocks. They execute the same deterministic algorithm. When moving, agents travel at the same speed normalized to 1. When agents are at distance at most $\epsilon$, for some positive constant $\epsilon$ unknown to them, they see each other and can exchange all information known to date.

Due to the anonymity of the agents and the symmetry of the plane, gathering is impossible, e.g.,  if agents start simultaneously at distances larger than $\epsilon$. However, if some agents start with a delay with respect to others, gathering may become possible. In which situations such latecomers can enable gathering? To answer this question we consider initial configurations formalized as sets of pairs  $\{(p_1,t_1), (p_2,t_2),\dots , (p_n,t_n)\}$, for $n\geq 2$ where $p_i$ is the starting point of the $i$-th agent and $t_i$ is its starting time. An initial configuration is {\em gatherable} if agents starting at 
it can be gathered by some algorithm, even dedicated to this particular configuration. 
Our first result is a characterization of all gatherable initial configurations. It is then natural to ask if there is a universal deterministic algorithm that can gather all gatherable configurations of a given size. It turns out that the answer to this question is negative. Indeed, we show that all gatherable configurations can be partitioned into two sets: {\em bad} configurations and {\em good} configurations. We show that bad gatherable configurations (even of size 2) cannot be gathered by a common gathering algorithm. On the other hand, we prove that there is a universal algorithm that gathers all good configurations of a given size. 

 Then we ask the question of whether the exact knowledge of the number of agents is necessary to gather all good configurations. It turns out that the answer is no, and we prove a necessary and sufficient condition on the knowledge concerning the number of agents that an  algorithm gathering all good configurations must have.  

\vspace*{0.5cm}

\noindent
{\bf keywords:}  anonymous agent, gathering, symmetry breaking, plane

\vspace*{2cm}
\end{abstract}

\setcounter{page}{0}
\pagebreak

\section{Introduction}
A team of anonymous mobile agents represented by points freely moving in the plane have to gather at a single point and stop, using the same deterministic algorithm. Agents start at
different points of the plane and at possibly different times chosen by the adversary. They are equipped with compasses, a common unit of distance and clocks. When moving, agents travel at the same speed normalized to 1. When agents are at distance at most $\epsilon$, for some positive constant $\epsilon$ unknown to them, they see each other and can exchange all information known to date.

 In applications, agents may be mobile robots moving in a terrain and taking samples of the ground in situations when it is impossible for humans to execute this task, due, e.g.,  to safety hazards. The task of gathering is a widely studied symmetry breaking task. Its importance is due both to practical and to theoretical reasons. On the practical side, mobile robots may have to gather to exchange collected samples or measurements and to coordinate further actions. On the theoretical side, gathering is important because (as justified below) it is equivalent
 to the fundamental symmetry breaking task of leader election among the agents. When is it feasible in a completely anonymous scenario of identical agents roaming in the empty plane?
 Can symmetry between the agents be broken by the delays between their starting times?
 
 \noindent
 {\bf The model and the problem.}
 Agents are modeled as anonymous points moving in the plane. Each agent appears at a different point of the plane at some time chosen by the adversary, and starts executing the same deterministic algorithm. Each agent has an accurate compass, a common unit of distance and a clock. 
 Agents do not know {\em a priori} the positions of other agents nor their appearance times. The origin of the coordinate system of each agent is at its starting point and its clock starts at the appearance time. Agents can execute instructions of two types: ``go in direction $dir$ at distance $x$'' and ``stay put for time $t$''. When an agent moves, it travels at speed 1.
 The execution of any of the above instructions is interrupted when an agent gets at distance at most $\epsilon$ from another agent, where $\epsilon>0$ is a constant {\em a priori} unknown to the agents. In this case, called an {\em approach}, an agent can see all agents at distance at most $\epsilon$ from it (i.e., it sees their coordinates in its own coordinate system) and can instantaneously exchange with them all information known to date. 
 Notice that an approach can ``suddenly'' happen for two agents at distance strictly smaller than $\epsilon$, if one  or both of them appear in the plane at some time, close to each other. Agents have unbounded memory and from the computational point of view they are modeled as Turing machines. After an approach, an agent may execute a new instruction, possibly based on the information learned during the approach. The goal of gathering is for all agents to get to the same point of the plane and stop forever. 
 
 The theoretical importance of the gathering problem is due to the fact that it is equivalent to the most fundamental symmetry breaking problem
 among anonymous agents in the plane, namely to leader election (cf. \cite{Ly}). Leader election calls for one agent to become a leader and all other agents to become non-leaders. To see the equivalence
 between these two problems, first suppose that leader election among agents is accomplished. Then the following simple algorithm achieves gathering. The leader stops forever, and all other agents explore the plane in phases indexed by natural numbers. In phase $i$, the agents explores the plane at distance $i$ from its original position 
 using a rectangular spiral with jump $1/i$; if it sees the leader in some phase, it joins it and stops, otherwise it goes back to its original position and starts the next phase. In some phase, each agent sees and joins the leader, which implies gathering. Conversely, if gathering is accomplished, then all agents exchange information concerning their starting points
relative to the coordinate system whose origin is at the gathering point, and  they choose as leader the agent whose starting point is the largest in some linear ordering (e.g., lexicographic order of coordinates).
 
An initial configuration of agents is formalized as a set of pairs  $\{(p_1,t_1), (p_2,t_2),\dots , (p_n,t_n)\}$, for $n\geq 2$, where $p_i$ is the starting point of the $i$-th agent in some global coordinate system and $t_i \geq 0$ is its starting time, according to some global clock. Recall that the agents do not have access to any global coordinate system or any global clock. An initial configuration is {\em gatherable} if agents starting at 
it can be gathered by some deterministic algorithm, even dedicated to this particular configuration. The main problems considered in this paper are the following.
\begin{quotation}
\noindent
Which initial configurations are gatherable?\\  Does there exist a universal deterministic algorithm gathering all of them?
\end{quotation}

\noindent
{\bf Our contribution.}
Our first result is a characterization of all gatherable initial configurations: these are configurations for which $|t_i-t_j|\geq dist(p_i,p_j)-\epsilon$, for at least one pair of agents $i $ and $j$, where $dist$ is the Euclidean distance in the plane. It is then natural to ask if there is a universal deterministic algorithm that can gather all gatherable configurations of a given size. It turns out that the answer to this question is negative. Indeed, we show that all gatherable configurations can be partitioned into two sets: {\em bad} configurations and {\em good} configurations.
Bad gatherable configurations are those for which $|t_i-t_j|\leq dist(p_i,p_j)-\epsilon$ for all pairs $i,j$, and good configurations are all other gatherable configurations.  We show that bad gatherable configurations (even of size 2) cannot be gathered by a common gathering algorithm. On the other hand, we design a universal algorithm that gathers all good configurations of a given size. 

Then we ask the question of whether the exact knowledge of the number of agents is necessary to gather all good configurations. It turns out that the answer is no, and we prove a necessary and sufficient condition on the knowledge concerning the number of agents that an  algorithm gathering all good configurations must have.  We also show that the assumption that agents see other agents when they are at some close but positive distance, is necessary for gathering.

In view of the equivalence between gathering and leader election, our results completely solve the general symmetry breaking problem
between anonymous mobile agents in the plane.

\noindent
{\bf Related work.}
The problem of gathering mobile agents, also called rendezvous, was widely studied in the literature.
Models under which it was investigated can be classified along two main dichotomies. The first dichotomy concerns the way in which agents move: it can be either deterministic or randomized. The second dichotomy concerns the environment in which agents navigate: it can be a network modeled as a graph or a terrain modeled as the plane, possibly with obstacles.

An excellent survey of  randomized gathering in various models  can be found in
\cite{alpern02b}, cf. also  \cite{alpern95a,alpern02a,anderson90}. 
Deterministic gathering in networks was surveyed in \cite{Pe}.
Gathering many labeled agents in the presence of Byzantine agents was studied in \cite{BDD,DPP}. 
The gathering problem was also studied in the context of oblivious robot systems in the plane, cf.
\cite{CP05,FPSW}, and fault tolerant gathering of robots in the plane was studied, e.g., in \cite{AP06,CP08}. 

Deterministic gathering in graphs with agents equipped with tokens used to mark nodes was considered, e.g., in~\cite{KKSS}. Deterministic gathering of two agents with unique labels was discussed in \cite{DFKP,KM,TSZ}.
These papers considered the time of gathering in arbitrary graphs. 
 In \cite{CKP,FP} the optimization criterion for gathering was the memory size of the agents:
it was studied in \cite{FP} for trees and in  \cite{CKP} for general graphs.
Memory needed for randomized gathering in the ring was discussed, e.g., in~\cite{KKPM08}. 

Apart from the synchronous model used in this paper, several authors considered asynchronous gathering in the plane \cite{BBDDP,CFPS,FPSW} and in networks 
\cite{BCGIL,CLP,DGKKP,DPV,GP}. In \cite{CFPS,FPSW} agents were anonymous, but were assumed to have total or restricted capability of seeing other agents.  However,  in the latter scenario only connected initial configurations were discussed. In \cite{BBDDP}, a related task of approach of two agents at distance 1 was investigated and agents were assumed to have distinct integer labels, which enabled them to break symmetry. 

Computational tasks in anonymous networks were studied in the literature, starting with the seminal paper \cite{A},
followed, e.g., by \cite{ASW, BV,KKV}. While the considered tasks, such as leader election in message passing networks or computing Boolean functions, differ
from gathering studied in the present paper, the main concern is usually symmetry breaking, similarly as in our case.

Deterministic rendezvous of anonymous agents in arbitrary anonymous graphs was previously studied in \cite{CKP,DP1,GP,PY}. Papers \cite{CKP,DP1} were concerned with the synchronous scenario. The main result of \cite{CKP} was a rendezvous algorithm working for all nonsymmetric initial positions using memory logarithmic in the size of the graph. \cite{DP1} was concerned with gathering multiple anonymous agents and characterized initial positions that allow gathering with all starting times. The authors of \cite{GP} characterized initial positions that allow asynchronous rendezvous. 
In a recent paper \cite{PY}, we considered the problem  of synchronous gathering of two anonymous agents in a graph. Similarly as in the present paper, we used time to break symmetry between the agents, but the situation was very different from our present setting. While
in \cite{PY} we showed a universal algorithm that gathers all initial configurations that can be gathered by a dedicated algorithm, in the present paper we show that such an algorithm for gathering in the plane cannot exist.  

\section{Preliminaries}

In the description of our algorithms we will use the notion of {\em generalized approach} (abbreviated as GA) of agents $a_1,\dots, a_k$. This is an event happening at some time $t$ such that immediately before this time (i.e., in the time interval $[t',t)$ for some $t'<t$), for  each pair of agents $a_i \neq a_j$, for $i,j=1,\dots, k$, these agents are either at distance larger that $\epsilon$ from each other, or at least one of them has not yet appeared at its initial position, and at time $t$, the agents  $a_1,\dots, a_k$ form a connected graph, where two agents are adjacent if they are at distance at most $\epsilon$. The simplest example of a GA is the approach of two agents that get at distance $\epsilon$. Another case of a GA is when there are two agents, one of which appears in the plane at time $t$ at distance smaller than $\epsilon$ from the other agent.  However,  more complicated cases of GA are also possible. Consider two agents $a_1$ and $a_2$ inert in points $p_1$ and $p_2$ at distance $2\epsilon$, and an agent $a_3$ travelling on the line perpendicular to the segment $p_1p_2$ crossing it in the middle $c$ of this segment. When the agent $a_3$ gets to the point $c$, a GA between $a_1$,  $a_2$ and $a_3$ occurs.
Notice that  adjacent agents in the  graph resulting from a GA can see each other and can exchange information. Hence full information exchange (gossiping) can be performed in this graph and all agents participating in the GA get to know the current positions and the trajectories of all other agents in the GA. According to our model, this information exchange is performed instantaneously.

Recall that each agent $x$ has its original system $\Sigma(x)$ of coordinates, centered at its initial position.
When an agent $a$ tells agent $b$ about the initial position $p$ of agent $c$ that it learned previously, it gives to agent $b$ the coordinates
of $p$ in the system of coordinates centered at the current position of $a$. The agent $b$, knowing the current position of $a$ in the system $\Sigma(b)$ (this is what it means that agent $b$ ``sees'' $a$)  translates these coordinates to $\Sigma(b)$. In this way, every agent $b$ can give a permanent label to each agent $c$ that it learned about, this label being the coordinates of $c$ in $\Sigma(b)$. All these labels are different.  

For any points $p,q$ in the plane we denote by $dist(p,q)$ the Euclidean distance between them.
Consider two points $p_1$ and $p_2$, such that $p_1$ has coordinates  $(x_1,y_1)$ and $p_2$ has coordinates  $(x_2,y_2)$ in some common system of coordinates. We will say that point $p_1$ is larger than $p_2$ (noted $p_1>p_2$) if  either $x_1>x_2$ or ($x_1=x_2$ and $y_1>y_2$).
We will say that an agent with starting position $p_1$ is larger than an agent with starting position $p_2$ if $p_1>p_2$. Similarly, we define an order between two vectors $v$ and $w$ as the lexicographic order on their coordinates. Notice that the notion of a point or a vector being larger  than another does not depend on the common system of coordinates used to define this relation, and hence all agents agree on this notion.

\section{Characterization of gatherable configurations}

In this section we characterize all gatherable initial configurations. The characterization is given by the following theorem.

\begin{theorem}\label{char}
An  initial configuration $\{(p_1,t_1), (p_2,t_2),\dots,(p_n,t_n)\}$ is gatherable if and only if $|t_i-t_j| \geq  dist(p_i,p_j)-\epsilon$, for some pair $i \neq j$ of agents.
\end{theorem}

The following lemma proves the ``only if'' implication of the equivalence from Theorem \ref{char}.

\begin{lemma}\label{neg}
An  initial configuration $\{(p_1,t_1), (p_2,t_2),\dots,(p_n,t_n)\}$ is not gatherable if $|t_i-t_j| <  dist(p_i,p_j)-\epsilon$, for all pairs $i \neq j$ of agents.
\end{lemma}

\begin{proof}
We prove the lemma by contradiction. Consider an initial configuration $\{(p_1,t_1)$,$\dots$,$(p_n,t_n)\}$ such that $\forall$ $i \neq j$, $|t_i-t_j| < dist(p_i,p_j)-\epsilon$. 
Suppose that, using some algorithm $\cal A$, gathering can be accomplished for this configuration. In order to get a contradiction, it is enough to show that no two agents $i\neq j$ can approach using algorithm $\cal A$. Suppose that some agents can approach and let $t$ be the earliest time of any approach. Let $i$ and $j$ be the agents that approach at time $t$. 
Without loss of generality, let $j$ be the later of them.  Let $\delta=t_j-t_i$. 

Since agent $i$ starts $\delta$ time ahead of agent $j$ and the agents are executing the same deterministic gathering algorithm $\cal A$, the route traversed by agent $j$ until approach is a shift of the route of agent $i$ until time $t-\delta$. 
So, the distance between the position of agent $i$ at time $t-\delta$ and the position of agent $j$ at time $t$ is $dist(p_i,p_j)$. Call this distance $d$. To achieve approach, the agent $i$ should cover distance at least $d-\epsilon$ during the last time segment of length $\delta$ before approach. However, by our assumption we have $\delta < d-\epsilon$, which gives a contradiction.
\end{proof}

In order to prove the ``if'' implication of the equivalence from Theorem \ref{char}, we design a gathering algorithm dedicated to a particular configuration satisfying the condition from the theorem.
Consider an initial configuration $\{(p_1,t_1), (p_2,t_2),\dots,(p_n,t_n)\}$ such that $|t_i-t_j| \geq  dist(p_i,p_j)-\epsilon$, for some pair $i \neq j$ of agents. This configuration remains fixed in the rest of this section and is given as input to the algorithm.

\vspace*{0.5cm}
\noindent
{\bf Algorithm} $DEDICATED$ 

The high-level idea of the algorithm is the following. In the first part, all agents make two moves that guarantee at least one GA. In the second part, the agents behave as follows.  If an agent  did not participate in a GA during the first part, it freezes and waits until it learns the initial positions of all other agents.  If an agent participated in a GA during the first part, an election mechanism is launched. As a result of it,  the agent either behaves as above, i.e., as if it did not participate in a GA, or becomes {\tt active} and starts visiting all possible initial positions of other agents in a round robin fashion (recall  that the agents know the initial configuration but they do not know their own position in it), at each GA exchanging all known information about the initial positions of other agents. Once an agent learns the initial positions of all other agents, it goes to the initial position of the largest agent and stops.

We now give a detailed description of the algorithm. At each time each agent is in one of the three states {\tt beginner}, {\tt active} or {\tt passive}. Below we describe the actions of an agent $A$ in each of these states. Let $V$ be the sequence of all vectors $(\overrightarrow{p_ip_j}: i,j \leq n, i\neq j)$ listed in increasing order.

\vspace*{0.5cm}
\noindent
State {\tt beginner}

Agent $A$ starts in this state. Let $v$ be the largest vector $\overrightarrow{p_ip_j}$ in $V$, such that $ |t_i-t_j| \geq  dist(p_i,p_j)-\epsilon$. Agent $A$ moves along the vector $v$ and backtracks to its initial position. If during these moves it did not participate in a GA, it transits to state {\tt passive}. If during these moves it participated in a GA, agent $A$ transits to state {\tt active} if it is the largest of all agents participating in its first GA, otherwise it transits to state {\tt passive}. 

\vspace*{0.5cm}
\noindent
State {\tt passive}

Agent $A$ waits at its initial position and during each GA it exchanges all available information about the initial positions of agents, until it learns the initial positions of all other agents. Then it goes
to the initial position of the largest agent and stops.

\vspace*{0.5cm}
\noindent
State {\tt active}

For consecutive vectors $w$ from the sequence $V$, agent $A$ goes along vector $w$ and backtracks to its initial position. After finishing the sequence $V$ it starts over again. During each GA
 it exchanges all available information about the initial positions of agents, and it travels along vectors from $V$ until it learns the initial positions of all other agents. Then it makes one more round going along all vectors of $V$ and informing all encountered {\tt passive} agents about the initial positions of all other agents. Finally, it goes
to the initial position of the largest agent and stops.

The proof of the correctness of Algorithm $DEDICATED$  is split into three lemmas.

\begin{lemma}\label{GA}
By the time when all agents make the first two moves along the vector $v$ and back, a GA must occur.
\end{lemma}

\begin{proof}
Consider two agents, $i$ and $j$, such that $v=\overrightarrow{p_ip_j}$. In state {\tt beginner} all agents move along the vector $v$ and back. If $dist(p_i,p_j) \leq \epsilon$, then these
two agents will wait at their initial positions after these moves, if they did not approach before, and the approach happens at the latest when they are back at their initial positions. Hence we may assume that
 $dist(p_i,p_j) > \epsilon$. Hence $|t_i-t_j|>0$. Consider two cases. If $t_i < t_j$, then by the time when agent $i$ travelling along the vector $\overrightarrow{p_ip_j}$ gets at distance $\epsilon$ from $p_j$, agent $j$ is still at its initial position $p_j$, because $t_j-t_i \geq  dist(p_i,p_j)-\epsilon$, and an approach between them happens. If $t_i>t_j$, let $\delta=t_i-t_j$. In this case, agent $j$ starts earlier. At time $dist(p_i,p_j)+\delta$ from its start, agent $j$ went from $p_j$ along the vector $\overrightarrow{p_ip_j}$ and either made the move back to $p_j$ (if $\delta\geq dist(p_i,p_j)$ ), or made a partial move towards $p_j$   getting at distance $dist(p_i,p_j)-\delta$
 (if $\delta< dist(p_i,p_j)$ ). At time $dist(p_i,p_j)+\delta$ from the start of agent $j$, agent $i$ made its move along the vector $\overrightarrow{p_ip_j}$, getting to $p_j$. In both cases, agents $i$ and $j$ are at distance at most $\epsilon$ at this time because $ t_i-t_j \geq  dist(p_i,p_j)-\epsilon$. Hence an approach between them happens.
\end{proof}

\begin{lemma}\label{know}
At some time of the execution of Algorithm $DEDICATED$ every agent knows the initial positions of all agents.
\end{lemma}

\begin{proof}
By Lemma \ref{GA}, a GA occurs at some time of the execution of Algorithm $DEDICATED$. 
One of the participants of a GA transits to state {\tt active} and all others transit to state {\tt passive}.
After the first two moves along the vector $v$ and back, every agent in state {\tt beginner} transits either to state {\tt passive} or to state {\tt active}. Hence, at some time $t$ of the execution, there are no agents in state {\tt beginner} and there is at least one agent in state {\tt active}
and at least one agent in state {\tt passive}. After time $t$,  any {\tt active} agent visits periodically every {\tt passive} agent until 
the latter learns the initial positions of all agents. Suppose that no {\tt passive} agent has learned initial positions of all agents yet. 
Pick any {\tt passive} agent $P$. All {\tt active} agents visit $P$ because $P$ does not move before learning all initial positions. Hence $P$ will learn the initial positions of all {\tt active} agents (either directly from the agent, or indirectly because another {\tt passive} agent $P'$  was visited by an agent $A'$, learned its initial position and then an {\tt active} agent $A$ visited $P'$ subsequently, learned the initial position of $A'$ from it and informed $P$ about it at the next visit). Also, an {\tt active} agent must visit any other
{\tt passive} agent $P'$, for the same reason. Hence, at the next visit at $P$ an {\tt active} agent will inform $P$ of the initial position of $P'$.
The last round of visits of an {\tt active} agent guarantees that an active agent will still visit all {\tt passive} agents after learning the initial positions of all of them, hence any {\tt passive} agent will learn the initial position of any other {\tt passive} agent.
It follows that $P$ will eventually learn the initial positions of all other agents and will go to the initial position $q$ of the largest agent.
Some {\tt active} agent $A$ may not find $P$ at its initial position because $P$ already started its travel to $q$. However, agent $A$ will eventually
visit $q$ after $P$ got there, and then it will learn the initial positions of all other agents, if it did not know them before.
\end{proof}

\begin{lemma}\label{pos}
Algorithm $DEDICATED$ gathers an initial configuration $\{(p_1,t_1), (p_2,t_2),\dots,(p_n,t_n)\}$ if  $|t_i-t_j| \geq  dist(p_i,p_j)-\epsilon$, for some pair $i \neq j$ of agents.
\end{lemma}

\begin{proof}
By Lemma \ref{know}, every agent eventually learns the initial positions of all other agents, if the condition ``$|t_i-t_j| \geq  dist(p_i,p_j)-\epsilon$, for some pair $i \neq j$ of agents'' is satisfied. Then every agent goes to the same initial position
$q$ of the largest agent and stops. This implies gathering of the initial configuration $\{(p_1,t_1), (p_2,t_2),\dots,(p_n,t_n)\}$.
\end{proof}

Now Theorem \ref{char} is a direct consequence of Lemmas \ref{neg} and \ref{pos}.

\section{Universal gathering}

Having characterized all gatherable initial configurations, we now study our second main problem: Does there exist a universal deterministic algorithm gathering all of them? It turns out that the answer to this question is negative. We solve this problem by classifying all gatherable initial configurations into the following two categories. {\em Bad} configurations are those gatherable initial configurations
$\{(p_1,t_1), (p_2,t_2),\dots,(p_n,t_n)\}$ for which $|t_i-t_j| \leq  dist(p_i,p_j)-\epsilon$, for all pairs $i \neq j$ of agents. {\em Good} configurations are all other gatherable initial configurations, i.e., those initial configurations 
$\{(p_1,t_1), (p_2,t_2),\dots,(p_n,t_n)\}$ for which $|t_i-t_j| >  dist(p_i,p_j)-\epsilon$, for some pair $i \neq j$ of agents. We will show that there is no universal deterministic gathering algorithm for the class of all bad configurations of a given size, but there is such a universal algorithm for the class of all good gatherable configurations of a given size.

The following proposition implies that there is no universal algorithm gathering all bad gatherable configurations of a given size. It even shows a stronger result, that achieving approach is impossible for this class of configurations.

\begin{proposition}\label{bad}
There does not exist a deterministic algorithm that guarantees an approach for all bad gatherable configurations of size 2.
\end{proposition}

\begin{proof}
Consider the problem of approach for two agents, where $\epsilon=1/2$. In order to get a contradiction, consider an algorithm $\cal A$
that guarantees approach for all bad gatherable configuration of two agents, for this $\epsilon$.
Consider the set $Z$ of initial configurations $\{(p_1,t_1),(p_2,t_2)\}$
such that $dist(p_1,p_2)=1$ and $|t_1-t_2|=1/2$. By Theorem \ref{char}, each configuration in the set $Z$ is gatherable, and by definition each of them is bad.
Hence algorithm $\cal A$ should gather all initial configurations from the set $Z$. Fix a point $x$ in the plane and, without loss of generality, suppose that the earlier agent is that starting at point $x$. All initial configurations $\{(x,0),(p,1/2)\}$, for all points $p$ at distance 1 from $x$, are in $Z$.  For any point $p$ at distance 1 from $x$, let $S(p)$ be the segment $xp'$, where $p'$ is the middle of the segment $xp$. Suppose that, for the configuration $\{(x,0),(p,1/2)\}$, the first approach accomplished by algorithm $\cal A$ is at time $t(p)$ (all times are counted from time 0). Hence, the segment between the position of the earlier agent at time  $t(p)-1/2$ and the position of the later agent at time $t(p)$ is a shift of the segment $S(p)$. In order to accomplish gathering, the earlier agent must traverse this segment during the time interval $[t(p)-1/2,t(p)]$. Hence the polygonal line $P$ resulting from the application of algorithm $\cal A$ at point $x$ must contain as segments shifts of all segments $S(p)$, where $p$ is at distance 1 from $x$. Any pair of such segments can intersect at most in one point. Suppose that these shifts in the polygonal line $P$ are in order $I_1, I_2,\dots$, 
corresponding to segments $S(p_1),S(p_2), \dots$. Since the set of points at distance 1 from $x$ is uncountable, there exists a point $p$ at distance 1 from $x$ that is different from all points $p_i$, $i=1,2,\dots$. Hence the polygonal line $P$  does not contain a shift of the segment $S(p)$, which is a contradiction.
\end{proof}

In contrast to the class of bad  gatherable configurations, we will show that all good gatherable configurations of a given size $n$ can be gathered by a universal algorithm $GATHER(n)$ whose initial knowledge is only the integer $n$.
Our algorithm will use the procedure $Star$ whose high-level idea is the following. In consecutive phases $x=1,2,\dots$,
the agent gets at distance at most $\frac{1}{x}$ from every point of the circle of radius $x$ centered at the starting position of the agent.  Each phase $x$ is executed in $k$ stages. Each stage corresponds to a ray of length $x$ from the starting position of the agent. The number $k$ of stages in a phase $x$ is $k = \lceil \frac{2\pi}{\alpha}\rceil$, where $\alpha$ is the angle between two rays of length $x$ whose endpoints are at distance $\frac{1}{x}$. The angle $\alpha$ is computed as follows:

\begin{equation}
 \sin (\alpha/2)=\dfrac{1}{2x^2}
\end{equation}
\begin{equation}
\alpha=2\cdot \arcsin\left(\frac{1}{2x^2}\right)
\end{equation}

In the first stage of phase $x$, the agent goes North at distance $x$ from its starting position, then comes back and waits time $x$ at its starting position. Next, in each subsequent stage $i=2,3,\dots,k$, the agent chooses a ray at angle $(i-1)\cdot\alpha$ clockwise from North and goes at distance $x$ along this ray, then comes back and waits time $x$ at its starting position. Algorithm \ref{star} gives a pseudocode of procedure $Star$. The procedure $Star$ is interrupted as soon as a GA occurs.

\begin{algorithm}
\Begin
{
$x:=1$\\
Repeat forever\\
\Begin
{
$\alpha : = 2 \arcsin(\frac{1}{2x^2})$ \\
$k : = \lceil \frac{2\pi}{\alpha}\rceil$\\
\For{$i=1$ to $k$}
{
Choose the half-line $H$ at angle $\alpha_i=(i-1)\cdot\alpha$ clockwise from North\\ 
Go along $H$ at distance $x$ \\
Go back at distance $x$ \\
Wait for time $x$\\
}
 }
 $x:=x+1$
 }
\caption{Procedure $Star$}
\label{star}
\end{algorithm}

Now, we describe algorithm $GATHER(n)$ that solves the problem of gathering for any good gatherable initial configuration $\{(p_1,t_1), (p_2,t_2),\dots,(p_n,t_n)\}$. 

\vspace*{0.5cm}
\noindent
{\bf Algorithm} $GATHER(n)$
 
 In order to convey the high-level idea of the algorithm, we consider four states in which an agent can be: {\tt cruiser}, {\tt explorer}, {\tt token} and {\tt shadow}. An agent in each of these states has a particular goal to accomplish and Figure $\ref{state}$ gives a diagram of transitions between the states.

An agent $A$  wakes up in state {\tt cruiser} and starts executing procedure $Star$. The main goal in this state is to achieve a GA after which agent $A$ transits to one of the other states. Agents in states {\tt token} and {\tt shadow} freeze and wait for the order to go to the final gathering position. Every {\tt explorer} is associated with a set of {\tt tokens}. It continues procedure $Star$ and becomes a {\tt shadow } when finding a {\tt token} larger than any of its own {\tt tokens}. Ultimately, only one {\tt explorer} remains: it is the one whose {\tt token} is the largest of all {\tt tokens}.
Every explorer counts the encountered agents in states {\tt token} or {\tt shadow}. When the last {\tt explorer} realizes that it has seen $n-1$ other agents (in states {\tt token} or {\tt shadow}) then it gives orders to all agents to go to the final gathering position and goes there itself. The reason for discerning {\tt tokens} from {\tt shadows} (which perform exactly the same actions) is to avoid freezing all {\tt explorers}.
Since {\tt explorers} become {\tt shadows} by comparing encountered {\tt tokens} to their own {\tt tokens}, we can guarantee that exactly one {\tt explorer} will remain to the end and give the final gathering order.

We now describe in detail the actions of an agent $A$ in each of the four states {\tt cruiser}, {\tt explorer}, {\tt token} and {\tt shadow}.

\vspace*{0.5cm}
\noindent
State {\tt cruiser}

Agent $A$ wakes up in this state and starts executing Procedure $Star$ until it participates in a GA. There are three cases.

Case 1. The GA  involves no agents in states {\tt token} or {\tt cruiser} apart  from $A$.

In this case agent $A$ simply ignores the GA.

Case 2. The GA involves other agents in state {\tt cruiser} but no agents in state {\tt token}.

After gossiping in the graph induced by the GA, the largest agent $B$ in state {\tt cruiser}  changes its state to {\tt explorer} (and continues the execution of $Star$). All other agents in state {\tt cruiser}  participating in the GA interrupt the execution of Procedure $Star$ and transit to state {\tt token}.

Case 3. The GA involves at least one agent in state {\tt token}.

Agent $A$ interrupts the execution of Procedure $Star$ and transits to state {\tt shadow}.

\vspace*{0.5cm}
\noindent
State {\tt explorer}

Agent $A$ transited to this state after a GA in which some agents in state {\tt cruiser} transited to state {\tt token}. It considers all these agents as its {\tt tokens}.
Agent $A$ keeps track of all initial positions of encountered agents in states {\tt token} and {\tt shadow}.
Agent $A$ continues the execution of Procedure $Star$ until a GA satisfying one of the two following conditions:
\begin{enumerate}
\item
The total count of all agents  in states {\tt token} or {\tt shadow} seen to date is $n-1$
\item
The total count of all agents  in states {\tt token} or {\tt shadow} seen to date is  less than $n-1$ and one of the agents in state {\tt token} in the current GA is larger than all {\tt tokens} of $A$

\end{enumerate}

In the first case agent $A$ repeats the last executed phase of procedure $Star$ and at each GA orders all participating agents to go to the
initial position of the largest of all $n$ agents and stop. Then it goes to the
initial position of the largest of all $n$ agents and stops.

In the second case, agent $A$ transits to state {\tt shadow}.

\vspace*{0.5cm}
\noindent
State {\tt token}

Agent $A$ stays at the position in which it transited to this state until receiving an order. Then it executes the order.

\vspace*{0.5cm}
\noindent
State {\tt shadow}

Agent $A$ stays at the position in which it transited to this state until receiving an order. Then it executes the order.

\begin{figure}[tp]
\centering
\includegraphics[scale=0.65]{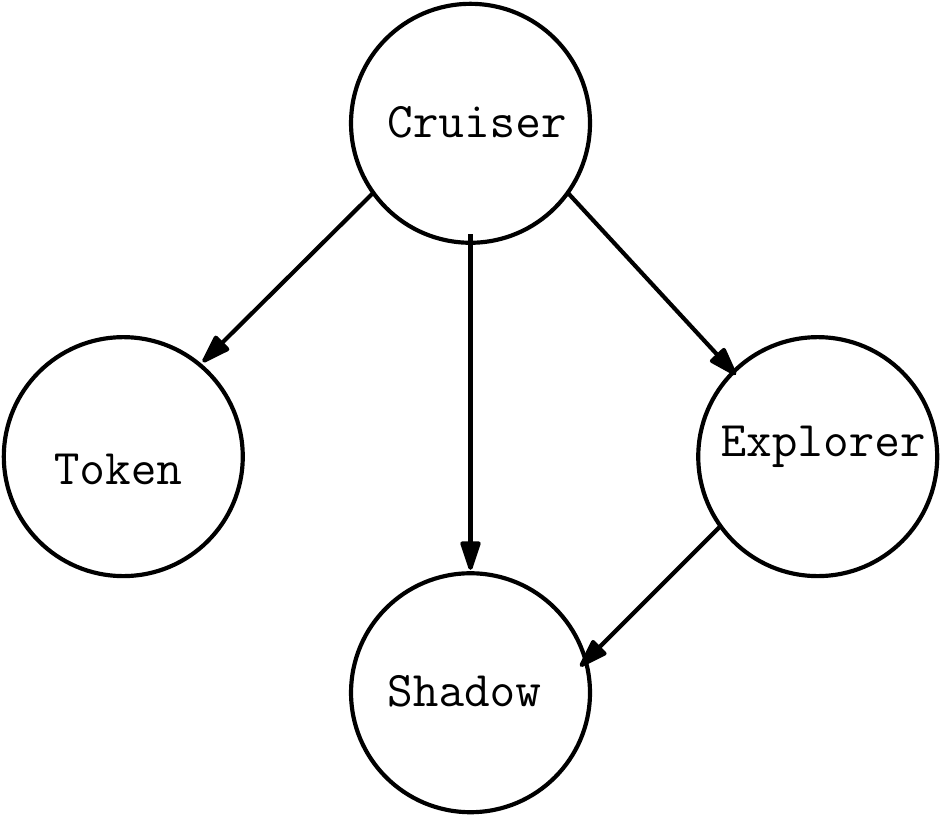}
\caption{State diagram of Algorithm $GATHER(n)$. }
\label{state}
\end{figure}


%
%
%
%
%
%
%

\vspace*{0.5cm}

The proof of correctness of Algorithm $GATHER(n)$ is split into several lemmas.
In all of them, we consider a good gatherable initial configuration $\{(p_1,t_1), (p_2,t_2),\dots,(p_n,t_n)\}$. 

%

\begin{lemma}\label{one-token}
At some time of the execution of Algorithm $GATHER(n)$,  at least one agent transits to state {\tt token}.
\end{lemma}

\begin{proof}
Since the configuration $\{(p_1,t_1), (p_2,t_2),\dots,(p_n,t_n)\}$ is good, there is a pair of agents $i,j$, such that 
$|t_i-t_j| > dist(p_i,p_j)-\epsilon$. Fix such a pair and assume,
without loss of generality, that $t_j>t_i$. Denote $\delta=t_j-t_i$ and $d=dist(p_i,p_j)$. 
Denote by $\beta$ the clockwise angle between the direction North and the line $p_ip_j$.
Let $z=\delta-(d-\epsilon)$. Hence $z>0$.
Let $x =\max( \lceil{ \delta} \rceil,\lceil1/z\rceil)$.
Let  $\alpha=2\cdot \arcsin\left(\frac{1}{2x^2}\right)$ be the angle between two rays of length $x$ in phase $x$ of Procedure $Star$, whose endpoints are at distance~$1/x$. We will prove the following claim.

{\bf Claim.} At some time during the execution of Procedure $Star$ by agent $j$, some 
GA  will happen (and hence some agent will transit to state {\tt token}).


In order to prove the claim,  recall that each stage of phase $x$ in procedure $Star$ takes time $3x$. Let $s=\left \lceil{\frac{\beta}{\alpha}}\right \rceil + 1$. Let $E_1$ be the execution of the $(s-1)$th stage of phase $x$ of Procedure $Star$ during the execution of Algorithm $GATHER(n)$ by agent $i$, and let $E_2$ be the execution of the $(s-1)$th stage of phase $x$ of Procedure $Star$ during the execution of Algorithm $GATHER(n)$ by agent $j$. The execution $E_1$ starts $\delta$ time ahead of the execution $E_2$.

 In the rest of the proof we assume that no GA happens by the end of the execution $E_2$; otherwise the claim is proved. 
 Let $t$ be the time when agent $j$ ends the execution $E_2$.
 Hence, agent $j$ waits at its initial position $p_j$ during the time interval $[t-x,t]$. Agent $i$ ends execution $E_1$ at time $t-\delta$, and starts the $s$th stage of phase $x$ moving along the half-line $L$ starting at its initial position $p_i$. It reaches the point $q$ at distance $d$ from $p_i$ on $L$ at time $t'=t-\delta+d= t+\epsilon-z$,  cf. Figure \ref{geo}. 
 
In phase $x$ of Procedure $Star$, the parameter $\alpha$ is chosen so that the distance between any two points, each located at distance $d' \leq x$ from $p_i$ on  rays with angle $\alpha$ between them is always at most $z$. Now, we show that an approach between agents $i$ and $j$ will happen by the time $t'$.  We consider the following two cases.

\textbf{Case 1.} $\epsilon \geq z$.\\
 Since agent $j$ waits at point $p_j$ until time $t$,  at time $t'=t+\epsilon-z$, agent $j$ is at some point $r$ at distance at most $\epsilon-z$ from point $p_j$. Since $dist(p_j,q) \leq z$, and $dist (r,p_j)\leq \epsilon -z$, we have $dist(q,r) \leq \epsilon$. This implies an approach between agents $i$ and $j$ by the time $t'$.

\textbf{Case 2.} $\epsilon < z$.\\ 
We consider two subcases and show that, in both of them, agent $i$ reaches point $q$ at the time when agent $j$ waits at its initial position $p_j$.

\begin{enumerate}
\item If $z\geq 1$ then $x =\max( \lceil{ \delta} \rceil,1)$. There are two possibilities. If $\lceil{\delta}\rceil \geq 1$ then $x=\lceil{\delta}\rceil$. This implies $t-x < t-x+d \leq t-\delta+d=t+\epsilon-z$. Since $\epsilon<z$, we have $t+\epsilon-z <t$. Hence, agent $i$ reaches point $q$ at the time when agent $j$ waits at its initial position. If  $\lceil{\delta}\rceil < 1$ then $x=1$. This implies $t-x=t-1 \leq t-\delta < t-\delta+d = t+\epsilon-z < t$. In this case also, agent $i$ reaches point $q$ at the time when agent $j$ waits at its initial position.
\item If $z<1$ then $x \geq 1$. This implies $t-x \leq t-1 <t-z < t+\epsilon-z < t$. So, agent $i$ reaches point $q$ at the time when agent $j$ waits at its initial position.  
\end{enumerate} 
Thus, agents $i$ and $j$ approach at the latest at time $t'$. This proves the claim and completes the proof of the lemma.
 \end{proof}

\begin{figure}[tp]
\centering
\includegraphics[scale=0.75]{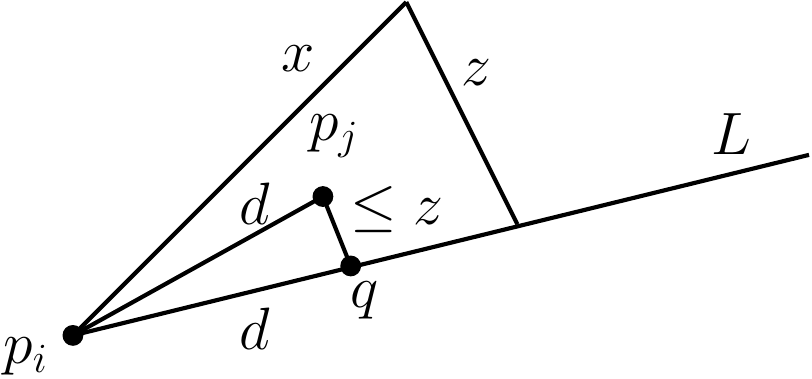}
\caption{Positions of agents $i$ and $j$ at time $t-\delta+d$.}
\label{geo}
\end{figure}

\begin{lemma}\label{no-cruiser}
At some time of the execution of Algorithm $GATHER(n)$,  no agent is in state {\tt cruiser}.
\end{lemma}

\begin{proof}
By Lemma \ref{one-token}, at some time of the execution of Algorithm $GATHER(n)$, there will be at least one agent $A$ in state {\tt token}. Let $u$ be the position of this agent at the time when it transits to state {\tt token}.
Any agent in state {\tt cruiser} transits to other states when it gets involved in a GA with agents in state {\tt cruiser} or {\tt token}.
Consider an agent $B$ in state {\tt cruiser}.
Let $y=\max(\lceil f\rceil, \lceil1/\epsilon\rceil)$ where $f$ is the distance between the initial position of the {\tt cruiser} $B$ and point $u$.
 
It is guaranteed that agent $B$ will participate in at least one GA  involving a {\tt cruiser} or a {\tt token}, at the latest  during the execution of phase $y$ of Procedure $Star$ by it. This is because, at the latest during the execution of this phase, agent $B$ approaches $A$.
Hence by the end of the execution of this phase, $B$ transits to some other state. This proves the lemma  
%
\end{proof}

\begin{lemma}\label{one-explorer}
At some time of the execution of Algorithm $GATHER(n)$, there is exactly one agent in state {\tt explorer}.
\end{lemma}

\begin{proof}
By Lemma \ref{no-cruiser}, at some time $t$ of the execution of Algorithm $GATHER(n)$, all agents are in one of the three states: {\tt explorer}, {\tt token} or {\tt shadow}, and no agent will transit to state {\tt token} anymore.  Let $X$ be the largest agent in state {\tt token} at time $t$, and
let $w$ be its position.
Let $D$ be the maximum distance between point $w$ and the position of any agent at time $t$.  


An agent $A$ in state {\tt explorer} executes Procedure $Star$ until it participates in a GA. Agent $A$ transits to state {\tt shadow} if the total count of all agents in states {\tt token} or {\tt shadow} seen to date by agent $A$ is less than $n-1$ and one of the agents in state {\tt token} in the current GA is larger than all {\tt tokens} of $A$. So, it is guaranteed that each agent in state {\tt explorer}, except the one associated with the {\tt token} $X$, will participate in at least one GA with some {\tt token} and transit to state {\tt shadow}, at the latest in phase
$\max(\lceil D\rceil, \lceil1/\epsilon\rceil)$ 
 of the Procedure $Star$ during the execution of Algorithm $GATHER(n)$. On the other hand, the {\tt explorer} associated with the {\tt token} $X$ will never change its state. 
This concludes the proof.
\end{proof}

%


\begin{lemma}\label{gathering}
At some time of the execution of Algorithm $GATHER(n)$, gathering is accomplished.
\end{lemma}

\begin{proof}
By Lemma \ref{one-explorer}, at some time $t'$ of the execution of Algorithm $GATHER(n)$, there is exactly one agent in state {\tt explorer}. Call this agent $A$. Let $D'$ be the maximum distance between agent $A$ and any other agent at time $t'$.
Agent $A$ will participate in GA's with each of the other $(n-1)$ agents at the latest in phase $\max(\lceil D'\rceil, \lceil1/\epsilon\rceil)$  of the Procedure $Star$ during the execution of Algorithm $GATHER(n)$ by $A$. At the latest during this phase, the total count of all agents  in states {\tt token} or {\tt shadow} seen to date by agent $A$ is $n-1$.  Then, agent $A$ repeats the last executed phase of Procedure $Star$ and at each GA orders all participating agents to go to the
initial position $g$ of the largest of all $n$ agents and stop. It subsequently goes to $g$ itself and stops. When all agents get to $g$ and stop, gathering is accomplished.  
%
%
\end{proof}

Lemma \ref{gathering} implies the following theorem.

\begin{theorem}\label{main}
Algorithm $GATHER(n)$ gathers all good gatherable initial configurations of size $n$.
\end{theorem}

\section{Discussion of assumptions}

In this section we discuss the necessity of two important assumptions  under which our gathering algorithm works. The first of them is the knowledge of the exact number of participating agents.
Can this assumption be removed? We will show that the exact knowledge of the number of agents is not necessary, but some partial knowledge about this number is needed. We will prove a necessary and sufficient condition on the knowledge that agents must have in order to gather all good gatherable initial configurations.

In order to formulate this condition, we formalize the initial knowledge of the agents as an {\em assumption set}, which is a set $A=\{a_1,a_2,\dots,a_k, \dots\}$ of integers, such that $1<a_1<a_2<\dots <a_k<\dots$. The set $A$ may be finite or infinite.
An assumption set is given to the agents in advance, and it conveys the following information: the size of the team of agents belongs to the set $A$. Thus the knowledge of the exact number of the agents that we used in the design of algorithm $GATHER(n)$ is a special case when the assumption set is a singleton.  We will say that a deterministic algorithm $\cal A$ working with the assumption set $A$ {\em solves the gathering problem}, if it gathers all good gatherable configurations of $n$ agents, for any $n \in A$. 

An assumption set $A$ is {\em dependent} if there exists an integer $a_k\in A$, a positive integer $r$ and positive integers $c_{i_1},c_{i_2},\dots,c_{i_r}$, for some $i_1<i_2<\dots <i_r<k$, such that $a_k=c_{i_1}a_{i_1}+c_{i_2}a_{i_2}+\cdots + c_{i_r}a_{i_r}$, i.e, if some element of $A$ is a positive linear combination of some smaller elements of it. An assumption set $A$ is {\em independent} if it is not dependent. There exist independent assumption sets of arbitrary large finite size,  for example the set $\{x,x+1,\dots,2x-1\}$ for any $x>1$, but every infinite set is dependent by Schur's theorem. Every singleton set is independent and, e.g.,  the sets $\{2,4\}$ and $\{2,3,7\}$ are dependent.
We will prove the following theorem.

\begin{theorem}\label{iff}
There exists an algorithm solving the gathering problem with the assumption set $A$ if and only if the set $A$ is independent.
\end{theorem}

The larger is the assumption set, the less precise is the knowledge about the number of agents, available to them. Of particular interest are assumption sets $\{2,3,\dots, x\}$ which formalize the assumption of a bound $x$ on the size of the team, given to the agents. These sets are independent only for $x=2$ and $x=3$. Apart from these cases, knowing a bound on the number of agents turns out to be insufficient information to solve the gathering algorithm. 

In order to prove Theorem \ref{iff}, we first prove the following lemma.

\begin{lemma}\label{dep}
If an assumption set $A$ is dependent, then there is no algorithm working with $A$ that solves the gathering problem.
\end{lemma}

\begin{proof}
Let $A=\{a_1,a_2,\dots,a_k, \dots\}$ be a dependent set, such that $a_k=c_{i_1}a_{i_1}+c_{i_2}a_{i_2}+\cdots +c_{i_r}a_{i_r}$, for some positive integers $c_{i_1},c_{i_2},\dots,c_{i_r}$.
Suppose that there exists an algorithm $\cal A$ that solves the gathering problem working with the assumption set $A$. For any $j\leq r$, let $I_{i_j}$ be a good gatherable initial configuration of size $a_{i_j}$.
Since $a_{i_j}$ belongs to $A$, algorithm $\cal A$ gathers it.
Consider the execution of algorithm $\cal A$ on the initial configuration $I_{i_j}$. In this execution, the trajectories of all agents of $I_{i_j}$ are within some disc $R_{i_j}$. 
Define a {\em shift} of  an initial configuration $I$ to be a configuration in which all points of configuration $I$ are translated by the same vector, and all starting times corresponding to the translated points are left unchanged. Thus, if an algorithm gathers an initial configuration, it also gathers any of its shifts.
Now consider an initial configuration  $C$ of size $a_k$ constructed in the following way. For any $j\leq r$, take $c_{i_j}$ shifts of configuration $I_{i_j}$, such that the respective discs containing trajectories of shifted configurations are at distance $2\epsilon$ from each other,  for different shifts of any configuration $I_{i_j}$, and for any shifts of any pair of configurations $I_{i_j}$ and $I_{i_{j'}}$, for $j\neq j'$. By definition, the initial configuration $C$ is good and its size belongs to $A$, so algorithm $\cal A$ should gather it. However, for any subconfiguration $C'$ of $C$ corresponding to any of the above described shifts, the execution of algorithm $\cal A$ on $C$ restriced to $C'$ must be identical to the execution of $\cal A$ on $C'$ because agents in this subconfiguration do not approach any agent outside it during the execution of algorithm $\cal A$ on configuration $C$. Hence, for all these subconfigurations, agents in each of them will gather at one point and stop. Since, by definition of $a_k$, there are at least two such subconfigurations in configuration $C$, there will be $c_{i_1}+c_{i_2}+\cdots + c_{i_r}\geq 2$ gathering points after reaching which by respective subsets of agents, no agent will ever move again. This contradicts the requirement that algorithm $\cal A$ should gather the initial configuration $C$.
\end{proof}

In order to prove the inverse implication of Theorem \ref{iff}, we indicate a gathering algorithm working  with any independent assumption set. Let $A=\{a_1,a_2,\dots,a_k\}$, where $1<a_1<a_2<\dots <a_k$, be an independent assumption set. The algorithm
$GATHER(A)$ working with this assumption set is a small modification of algorithm $GATHER(n)$ designed in the previous section.
The modification is at the very end of $GATHER(n)$, when the last remaining {\tt explorer} realizes that there are $n-1$ agents in states {\tt token} and {\tt shadow}, orders them to go to the gathering point and stop, and does the same itself. The modification is as follows. 

The agents start by executing algorithm $GATHER(a_1)$ (i.e., they ``assume'' that the total number of agents is $a_1$). Every agent in state
{\tt explorer} works with a current assumption which is an integer $a\in A$. The first assumption is $a_1$. When an {\tt explorer} realizes that the
total number of agents is larger than its current assumption $a$, it changes the assumption to the smallest number $a'>a$ in $A$ and continues algorithm $GATHER$ with this assumption, i.e., switches from executing $GATHER(a)$ to executing $GATHER(a')$. If this happens after a temporary gathering of a subset of agents, an  {\tt explorer} taking part in this temporary gathering continues executing procedure $Star$ from the next phase after it interrupted it. As opposed to algorithm $GATHER(n)$, {\tt tokens} and {\tt shadows} may
get several consecutive orders, corresponding to consecutive assumptions, which they execute one by one. At the end of the modified algorithm
$GATHER(A)$ all agents are grouped in one or more points and no agent moves anymore. The size of each group corresponds to some assumption, i.e., it is an integer from $A$, and the total number of agents is also an integer from $A$. Since the set $A$ is independent, 
there must be a single group gathered at one point, which implies that the gathering is accomplished. Thus we have the following proposition.

\begin{proposition}\label{ind}
For any independent assumption set $A$, algorithm $GATHER(A)$ gathers all good gatherable initial configurations of size belonging to $A$.
\end{proposition}

Theorem \ref{iff} is a direct consequence of Lemma \ref{dep} and Proposition \ref{ind}. 

There is a subtle difference between the capabilities of a gathering algorithm working with assumption set $A$  when $A$ is a singleton and when $A$ is an independent set of size larger than 1. Since in the first case all agents know the size $n$ of the team, after gathering they learn that gathering is accomplished, i.e., not only do they stop forever, but they can declare that the task is finished. This is sometimes called gathering with detection. This is,
however, impossible if the assumption set has more than 1 element. Suppose that the assumption set is $\{2,3\}$ but that the real number of agents is 2. Then the two agents will eventually gather and stop forever, but they can never declare that the gathering is accomplished because the real number of agents could be 3 and the third agent starting from a very distant position could be on its way trying to find another agent. This is actually what would happen for algorithm $GATHER(\{2,3\})$, for a good gatherable configuration of 3 agents two of which are close to each other, one is very distant and the waking times are pretty close: two agents would gather soon and temporarily stop, while the distant agent would eventually join them. Hence gathering with detection of all good gatherable configurations with assumption set $A$ is possible if $A$ is a singleton but it is (in general) impossible otherwise.

The second assumption that we used in the design of our gathering algorithm was the capability of the agents to see each other and communicate upon approach, i.e., when they get at distance at most $\epsilon$, where $\epsilon$ is some positive constant. Is this assumption necessary? In other words, does there exist an algorithm gathering all good gatherable initial configurations of known size in the model in which, before gathering, agents do not see each other and cannot interact? Such agents will be called {\em non-communicating}. For non-communicating agents, a good gatherable configuration is defined by the condition
``there is a pair of agents $i,j$, such that $|t_i-t_j| > dist(p_i,p_j)$'' because the communicating radius $\epsilon$ is 0 in this case.

The following result shows that such an algorithm for non-communicating agents does not exist (even for two agents), i.e., our assumption could not be removed.

\begin{proposition}
There does not exist a deterministic algorithm for non-communicating agents that gathers all good gatherable configurations of size 2.
\end{proposition}

\begin{proof}
Suppose that such an algorithm $\cal A$ does exist. Let $x$ be a fixed point in the plane, and fix the system of coordinates with origin at $x$. Let $Q=\{p: 0< dist(x,p) \leq 1\}$ and consider all initial configurations $\{(x,0),(p,2)\}$,
for $p\in Q$. All these configurations are good, hence algorithm $\cal A$ should gather all of them. Let $\phi$ be the trajectory resulting from the execution of algorithm $\cal A$ starting at point $x$, i.e., the polygonal line such that an agent executing $\cal A$ and starting at $x$ is at point $\phi(t)$ at time $t$. (For non-communicating agents, a gathering algorithm has no "if" statements, thus it simply produces a polygonal  line, following which the agents should meet at some point).
For any point $p\in Q$ and for the initial configuration $\{(x,0),(p,2)\}$, the agent starting at $p$ starts at time $2$. Let $\psi_p$ be the trajectory of this agent. Hence, for any $t\geq 2$, we have $\psi _p(t)=p+\phi(t-2)$, as the trajectory $\psi_p$ is a shift of trajectory $\phi$ be the vector $p$. If the agent starting at $x$ at time 0 and the agent starting at $p$ at time 2 meet at some time $t \geq 2$, we have $\phi(t)=\psi_p(t)$, i.e., $p=\phi(t)-\phi(t-2)$. Let $\Phi$ be  defined by the formula $\Phi(t)= \phi(t)-\phi(t-2)$, for $t\geq2$. The function $\Phi$ describes a polygonal line. Since, for any $p \in Q$, the agent starting at $x$ at time 0 and the agent starting at $p$ at time 2 must meet at some time $t \geq 2$,
the polygonal line $\Phi$ must cover the entire disc $Q$, i.e., for every $p \in Q$, there must exist a time $t \geq 2$, such that $\Phi(t)=p$. However, this is impossible because every polygonal line has two-dimensional Lebesgue measure 0 (as a union of countably many segments) and the disc $Q$ has two-dimensional Lebesgue measure $\pi$.
\end{proof}

\section{Conclusion}

We characterized all initial configurations of anonymous agents in the plane that can be gathered by a dedicated algorithm, and we answered the question if a universal algorithm can gather all such configurations of a given size. The answer is no, and we showed for which sets of gatherable configurations such a common gathering algorithm exists: we classified all gatherable configurations into two categories of bad and
of good gatherable configurations, and we showed that while all bad ones cannot be gathered by a common algorithm, all good ones can. Then we showed that the knowledge of the exact number of agents is not necessary to gather all good gatherable configurations. We proved a necessary and sufficient condition on the knowledge that an algorithm needs, in order to be able to gather all good gatherable configurations.

We also showed that if agents cannot interact at all  before gathering, then gathering by a common algorithm cannot be accomplished for all good gatherable configurations, even for two agents.
However, the following question remains open. In our model we assume that, upon approach, agents see each other and can communicate
exchanging all information known to date. Is this communication ability necessary? Do our results hold in a weaker model where agents at distance at most $\epsilon$ can see each other but cannot talk? Clearly, for teams of two agents this does not change the ability of gathering (when two agents see each other, there is nothing to talk about: the agent whose current position is larger stops, and the other agent joins it) but we do not know what is the answer for arbitrary teams of agents.


\end{document}